\documentclass[11pt]{article}

\usepackage[a4paper, left=3.2cm,top=3.2cm]{geometry}
\usepackage{authblk}
\usepackage[onehalfspacing]{setspace}

\usepackage[T1]{fontenc}

\usepackage[bookmarks=closed]{hyperref}
\usepackage{breakurl}

\usepackage{enumitem}

\usepackage{color}

\usepackage{amsfonts}
\usepackage{amsmath, bm, amssymb}
\usepackage{amsthm}
\usepackage{mathtools}

\newtheorem{theorem}{Theorem}
\newtheorem{lemma}[theorem]{Lemma}

\theoremstyle{definition}
\newtheorem{definition}[theorem]{Definition}
\newtheorem{remark}[theorem]{Remark}

\newtheorem{cond}{Condition}
\newtheorem{ass}{Assumption}

\usepackage{tikz}
\usepackage{graphicx}
\usepackage[font=small,labelfont=bf]{caption} 
\usepackage{float}
\usepackage{xcolor}

\DeclareMathOperator{\dom}{dom}

\newcommand{\Ao}{\mathbf{A}}
\newcommand{\Mo}{\mathbf{M}}
\newcommand{\Mall}{\mathbf{M}}
\newcommand{\Wo}{\mathbf{W}}

\newcommand{\N}{\mathbb{N}}
\newcommand{\R}{\mathbb{R}}

\newcommand{\relaxed}{\mathcal{B}}
\newcommand{\strict}{\mathcal{A}}
\newcommand{\reg}{\mathcal{R}}
\newcommand{\bregman}{\mathcal{D}}
\newcommand{\qreg}{\mathcal{Q}}

\newcommand{\la}{\lambda}
\newcommand{\La}{\Lambda}
\newcommand{\Om}{\Omega}
\newcommand{\ka}{\kappa}
\newcommand{\al}{\alpha}

\newcommand{\etaw}{\eta}

\newcommand{\data}{y}
\newcommand{\signal}{x}
\newcommand{\hid}{h}
\newcommand{\X}{\mathbb X}
\newcommand{\Y}{\mathbb Y}
\newcommand{\Hr}{\mathbb H}

\newcommand{\edot}{\,\cdot \, }

\newcommand{\supp}{\operatorname{supp}}
\newcommand{\Span}{\operatorname{span}}
\newcommand{\sign}{\operatorname{sign}}

\newcommand{\ran}{\operatorname{ran}}

\newcommand{\id}{\operatorname{Id}}
\newcommand{\argmin}{\operatorname{argmin}}

\newcommand\abs[1]{\left\vert#1\right\vert}

\newcommand\norm[1]{{\left\Vert#1\right\Vert}}
\newcommand\snorm[1]{\Vert#1\Vert}
\newcommand{\enorm}{{\left\Vert \cdot \right\Vert}}

\newcommand\inner[2]{\langle#1,#2\rangle}
\newcommand\sinner[2]{\langle#1,#2\rangle}

\usepackage{soul}
\usepackage{xcolor}
\colorlet{lred}{red!40}
\colorlet{lgreen}{green!40}
\colorlet{lblue}{blue!40}

\numberwithin{equation}{section}
\numberwithin{theorem}{section}

\allowdisplaybreaks

\title{Convergence rates for the joint solution of inverse problems with compressed sensing data}

\date{May 15, 2022}

\author{Andrea Ebner}

\affil{Department of Mathematics, University of Innsbruck\authorcr
Technikerstrasse 13, 6020 Innsbruck, Austria
 \authorcr E-mail:  \texttt{andrea.ebner@uibk.ac.at}
 }

\author{Markus Haltmeier}

\affil{Department of Mathematics, University of Innsbruck\authorcr
Technikerstrasse 13, 6020 Innsbruck, Austria
 \authorcr E-mail:  \texttt{markus.haltmeier@uibk.ac.at}
 }

\begin{document}

\maketitle

\begin{abstract}
Compressed sensing (CS) is a powerful tool for reducing the amount of data to be collected while maintaining high spatial resolution. Such techniques work well in practice and at the same time are supported by solid theory. Standard CS results assume measurements to be made directly on the targeted signal. In many practical applications, however, CS information can only be  taken from indirect data  $\hid_\star = \Wo \signal_\star$ related to the original signal by an additional forward operator. If inverting the forward operator is ill-posed, then existing CS theory is not applicable. In this paper, we address this issue and present two joint reconstruction approaches, namely relaxed $\ell^1$ co-regularization and strict $\ell^1$ co-regularization,  for  CS from indirect  data. As main results, we derive error estimates for recovering $\signal_\star$ and  $\hid_\star$.  In particular, we derive a linear convergence rate in the norm for the latter. To obtain these results, solutions are required to satisfy a source condition and the CS measurement operator is required to satisfy a restricted injectivity condition. We further show that these conditions are not only sufficient but even necessary to obtain linear convergence.


\end{abstract}

\section{Introduction}

Compressed sensing  (CS) allows to  significantly reduce the amount of  measurements while keeping high spatial resolution \cite{candes2006robust,donoho2006compressed,fourcat13CS}.  In mathematical  terms,  CS requires recovering a targeted   signal $\signal_\star \in \X$  from data  $\data^\delta = \Mo \signal_\star + z^\delta $. Here $\Mo \colon  \X \to \Y$ is the CS measurement operator, $\X$, $\Y$ are Hilbert spaces and $z^\delta \in \Y$ is the unknown data perturbation with $\snorm{z^\delta} \leq \delta$.  CS theory shows that  even when the measurement operator is severely under-determinated one can derive linear  error estimates   
$ \snorm{ \signal^\delta  - \signal_\star} = \mathcal O(\delta)  $ for the CS reconstruction $\signal^\delta$.  Such results can  be derived uniformly for all sparse $\signal_\star \in \X$  assuming  the restricted isometry property  (RIP) requiring that $\snorm{\Mo \signal_1 - \Mo \signal_2} \asymp \snorm{\signal_1 - \signal_2} $ for sufficiently sparse elements \cite{candes2006stable}. The  RIP is known to be satisfied  with high probability for a wide range of random matrices~\cite{baraniuk2008simple}.   Under a restricted injecticity  condition, related results for elements satisfying a range condition are derived  in \cite{Gr08,Gr11}. In \cite{fuchs2004sparse} a strong  form of the source condition has been shown to be sufficient and necessary for the uniqueness of $\ell^1$ minimization. In \cite{Gr11} it is shown that the RIP implies  the source condition and the restricted injectivity  for all sufficiently sparse elements.

\subsection{Problem formulation}
 
In many applications, CS measurements can only be made on indirect data  $\hid_\star = \Wo \signal_\star$ instead of the  targeted  signal $\signal_\star \in \X$, where $\Wo \colon  \X \to \Hr$ is the forward operator coming from a specific application at hand. For example, in  computed tomography,  the forward operator  is the Radon transform, and in  microscopy, the forward operator is  a convolution operator. The  problem of recovering  $\signal_\star \in \X$ from CS measurements of indirect data becomes 
\begin{equation}\label{eq:csi}
	\text{Recover $\signal_\star$ from } \data^\delta = \Ao \Wo \signal_\star + z^\delta \,,
\end{equation}
where $\Ao \colon \Hr \to \Y$ is the  CS measurement operator. In this paper we study the  stable  solution of \eqref{eq:csi}.   
 
The naive reconstruction approach is a single-step approach to consider \eqref{eq:csi} as standard CS problem with the  composite measurement operator $\Mo = \Ao \Wo$. However,  CS recovery conditions (such as the RIP) are not expected  to hold  for the composite operator $\Ao \Wo$ due to the ill-posedness of the operator $\Wo$.  As an alternative one may use a two-step approach where one first solves the CS problem of recovering $\Wo \signal_\star$ and afterwards inverts the operator equation of the inverse problem.  Apart from the additional effort, both recovery  problems need to be regularized and the risk of error propagation is high. Moreover, recovering $\hid_\star \in \Hr$ from sparsity alone suffers from increased non-uniqueness  if  $\overline{\ran(\Wo)} \subsetneq \Hr$.

\subsection{Proposed $\ell^1$ co-regularization}

In order to overcome the drawbacks of the single-step and the two-step approach, we introduce two joint reconstruction methods  for solving \eqref{eq:csi}  using a weighted $\ell^1$ norm $\snorm{\edot}_{1,\kappa}$ (defined in \eqref{eq:norm}) addressing  the CS part and variational regularization with an additional penalty $\reg$ for addressing the inverse problem part. More precisely, we study the  following two regularization approaches.       
\begin{enumerate}[label=(\alph*)]
\item \textsc{Strict  $\ell^1$ co-regularization:} 
Here we construct  a regularized solution pair $(\signal_\al^\delta, \hid_\al^\delta)$ with $\hid_\al^\delta = \Wo \signal_\al^\delta$ by minimizing       
\begin{equation}\label{eq:strict} 
	\strict_{\al,\data^\delta} (\signal) \coloneqq  \frac{1}{2} \snorm{\Ao\Wo \signal -\data^\delta}^2 
	+ \al \bigl( \reg(\signal) + \snorm{\Wo \signal}_{1,\kappa} \bigr)  \,,
\end{equation}
where $\al > 0$ is a regularization parameter. This is equivalent  to minimizing    
$ \snorm{\Ao \hid -\data^\delta}^2/2  + \al ( \reg(\signal) + \snorm{\hid}_{1,\kappa} )$ under the strict constraint $\hid = \Wo \signal$.
    
\item\textsc{Relaxed $\ell^1$ co-regularization:} 
Here we relax  the constraint $\hid = \Wo \signal$  by adding  a penalty and construct  a regularized solution $(\signal_\al^\delta, \hid_\al^\delta)$ by minimizing              
\begin{equation}\label{eq:relax}
\relaxed_{\al,\data^\delta} (\signal,\hid) \coloneqq \frac{1}{2} \snorm{\Wo \signal - \hid}^2 + \frac{1}{2} \snorm{ \Ao \hid - \data^\delta }^2  \\ 
+ \al \bigl(  \reg(\signal) + \snorm{\hid}_{1,\kappa} \bigr) \,.
\end{equation}
The relaxed version in particular allows some defect between $\Wo \signal_\al^\delta$ and $\hid_\al^\delta$.   
\end{enumerate}
Under standard  assumptions, both the strict and the relaxed version provide convergent regularization methods \cite{scherzer2009variational}.

\subsection{Main results}

As main results of this paper, under  the parameter choice $\alpha \asymp \delta$, we  derive the linear convergence rates  (see Theorems~\ref{thm:strict},~\ref{thm:relaxed}) 
\begin{align*}
\bregman_{\xi}^{\reg}  (\signal_\al^\delta, \signal_\star)  &= \mathcal O(\delta)  \quad \text{ as } \delta \to 0 \\
\snorm{ \hid_\al^\delta  -  \Wo \signal_\star} &= \mathcal O(\delta) \quad \text{ as } \delta \to 0 \,,
\end{align*} 
where  $\bregman_{\xi}^{\reg}$ is the Bregman distance with respect  to $\reg$ and $\xi$  (see Definition~\ref{def:bregman}) for  strict as well as for  relaxed $\ell^1$ co-regularization.  In order to archive these results, we assume a restricted injectivity condition for  $\Ao$ and  source conditions for  $\signal_\star$ and  $\Wo \signal_\star$. These above error estimates are optimal in the sense that  they cannot be improved  even  in the cases where $\Ao = \id$, which  corresponds to an inverse problem only, or the  case   $\Wo = \id$ where   \eqref{eq:csi} is a standard CS problem on direct data.   

As further main  result we  derive converse results, showing that the source condition and the restricted injectivity condition are also necessary to obtain linear convergence rates (see Theorem~\ref{thm:converse}).

We note  that our results and analysis closely follow  \cite{Gr08,Gr11}, where the  source condition and restricted injectivity are shown to be necessary and sufficient for linear convergence of  $\ell^1$-regularization for CS on direct data. In that context, one considers CS as a particular instance of an inverse problem under a sparsity prior using variational regularization with an $\ell^1$-penalty (that is, $\ell^1$-regularization). Error estimates in the norm distance for $\ell^1$-regularization  based on the source condition have been first derived in \cite{lorenz2008convergence} and strengthened  in \cite{Gr08}.  In the finite dimensional setting, the source condition (under a different name) for $\ell^1$-regularization has been used previously in \cite{fuchs2004sparse}. For some more recent development of $\ell^1$-regularization and source conditions see \cite{flemming2018variational}. 

Further note that for the direct CS problem where $\Wo=\id$  is the identity operator and if we take  the regularizer  $\reg=\snorm{\edot}^2/2$, then the strict $\ell^1$ co-regularization reduces to the well known elastic net  regression model \cite{Zo05}.  Closely following the work \cite{Gr08}, error estimates for elastic net  regularization have been derived in \cite{jin2009elastic}.  Finally, we note that another interesting line of research in the context of $\ell^1$ co-regularization would be the derivation of  error estimates under the RIP. While we expect  this to be possible, such an analysis is  beyond the  scope of this work.

\section{Linear convergence rates}\label{sec:rate}

Throughout this paper $\X, \Y$ and $\Hr$ denote separable Hilbert spaces with inner product  $\sinner{\edot}{\edot}$ and norm $\snorm{\edot}$.   Moreover we make the following assumptions.

\begin{ass} \mbox{} \label{ass:main}
\begin{enumerate}[label=(A.\arabic*),leftmargin=1cm]
\item\label{ass:main2} $\Wo \colon \X  \to  \Hr$ is linear and bounded. 
\item\label{ass:main1} $\Ao \colon \Hr \to \Y$ is   linear and bounded.
\item\label{ass:main3}  $\reg \colon \X \to [0,\infty]$ is  proper, convex and wlsc.
\item\label{ass:main4}  $\La$ is a countable index set.
\item\label{ass:main5}  $(\phi_\la)_{\la \in \La} \in \Hr^\La$ is an orthonormal basis (ONB) for $\Hr$.
\item\label{ass:main7}  $(\ka_\la)_{\la \in \La} \in [a, \infty)^\La$ for some $a > 0$.
\item\label{ass:main8}  $\exists \signal \in  \X \colon \reg(\signal) + \sum_{\la \in \La} \ka_\la \abs{\inner{\phi_\la}{\Wo\signal}} < \infty$.
\end{enumerate}
\end{ass}

Recall that $\reg$ is wlsc (weakly lower semi-continuous) if $ \liminf_{k \to \infty} \reg (\signal_k ) \geq \reg( \signal)$ for all     $(\signal_k)_{k\in \N} \in \X^\N$ weakly  converging to  $\signal \in \X$.  We write    $\ran(\Wo)  \coloneqq \{\Wo \signal \mid \signal \in \X \}$ for the range of $\Wo$ and   
$$\supp(\hid) \coloneqq \{\la \in \La \mid \inner{\phi_\la}{\hid} \neq 0 \} $$  for the support of $\hid \in \Hr$ with respect to  $(\phi_\la)_{\la \in \La}$. 
A signal $\hid  \in \Hr$ is sparse if $\abs{\supp(\hid)} < \infty$. 
The weighted $\ell^1$-norm $\snorm{\edot}_{1,\kappa} \colon \Hr \to [0, \infty]$ with  weights $(\ka_\la)_{\la \in \La} $ is defined by    
\begin{equation} \label{eq:norm}
	\snorm{\hid}_{1,\kappa}  
	 \coloneqq  \sum_{\la \in \La} \ka_\la \abs{\inner{\phi_\la}{\hid}} \,.
\end{equation}
We have $\dom(\snorm{\edot}_{1,\kappa}) = \{\hid \in \Hr \mid \sum_{\la \in \La} \ka_\la \abs{\inner{\phi_\la}{\hid}} < \infty \}$.
For a finite subset of indices $\Omega \subseteq \La$, we write  
\begin{align}
	&\Hr_\Om \coloneqq \Span\{\phi_\la \mid \la \in  \Omega \}  \\
	&\mathrm{i}_{\Omega} \colon \Hr_\Om \to \Hr \colon  h \mapsto h\\
	&\Ao_{\Omega} \coloneqq \Ao \circ \, \mathrm{i}_{\Omega} \colon \Hr_\Om \to \Y \\
	&\pi_{\Omega} \colon  \Hr \to \Hr_\Om  \colon \hid \mapsto \sum_{\la \in \Omega} \inner{\phi_\la}{x}\phi_\la  \,.
\end{align}
Because  $(\phi_\la)_{\la \in \La} \subseteq \Hr$ is an ONB,  every  $\hid \in \Hr$ has the basis representation $ \hid=\sum_{\la \in \La} \inner{\phi_\la}{\hid}\phi_\la$.  Finally,   $\norm{\Ao}$ denotes the standard operator norm.

\subsection{Auxiliary estimates} \label{sec:aux}

One main ingredient  for our results are error estimates for general variational regularization in terms of the Bregman distance.
Recall that  $\xi \in \X$  is called   subgradient  of a functional $\qreg \colon  \X \to [0, \infty]$ at $\signal_\star \in \X$ if 
$$
	 \forall \signal \in \X \colon \quad   
	 \qreg(\signal) \geq \qreg(\signal_\star) + \inner{\xi}{\signal - \signal_\star} \,.
$$
The set of all subgradients is called the subdifferential of $\qreg $ at $\signal_\star$ and denoted by $\partial \qreg(\signal_\star)$.  

\begin{definition}[Bregman distance] \label{def:bregman}
Given  $\qreg \colon  \X \to [0, \infty]$ and $\xi \in \partial \qreg(\signal_\star)$,  the  Bregman distance between $\signal_\star, \signal \in \X$  with respect  to $\qreg$ and $\xi$ is defined  by 
\begin{equation}
	\bregman_{\xi}^\qreg(\signal,\signal_\star) \coloneqq \qreg(\signal) - \qreg(\signal_\star) - \inner{\xi}{\signal-\signal_\star} \,.
\end{equation}
 \end{definition}
 
The Bregman distance is a valuable tool for deriving error estimates for variational regularization. Specifically, for  our purpose we use the following convergence rates result.

\begin{lemma}[Variational regularization] \label{lem:var}
Let  $\Mall \colon \X \to \Y$ be bounded and linear,  let $\qreg \colon  \X \to [0, \infty]$ be proper, convex and wlsc and let $(\signal_\star, \data_\star) \in \X \times \Y$ satisfy $\Mall \signal_\star = \data_\star$ and  $\Mall^* \eta \in \partial \qreg(\signal_\star)$ for some $\eta \in \Y$. Then  for all $\delta, \al > 0$, $\data^\delta \in \Y$ with  $\snorm{\data^\delta - \data_\star} \leq \delta$ and  $\signal_\al^\delta \in \argmin \{ \snorm{\Mall \signal - \data^\delta}^2/2 + \al \qreg(\signal)\}$ we have 
\begin{align} \label{eq:var1}
\snorm{\Mall \signal_\al^\delta - \data^\delta } & \leq \delta + 2 \al \snorm{\eta}   
\\ \label{eq:var2}
D_{\Mall^* \eta}^{\qreg} (\signal_\al^\delta,\signal_\star) & \leq  (\delta + \al \snorm{\eta})^2  / (2 \al)  \,.
\end{align}
\end{lemma}

\begin{proof}
Lemma \ref{lem:var} has been derived in   \cite[Lemma 3.5]{Gr11}. Note  that error estimates  for variational regularization  in the Bregman distance have first been derived in~\cite{burger2004convergence}.   
\end{proof}

For our purpose we  will apply Lemma \ref{lem:var} where $\qreg$ is a combination formed  by $\reg $ and $\snorm{\edot}_{1,\kappa}$.  We will use that the subdifferential of  $\snorm{\edot}_{1,\kappa}$ at $\hid_\star$ consists of  all $\eta = \sum_{\la \in \La} \eta_\la \phi_\la \in \Hr$ with  
\begin{equation*}
	\begin{cases}
	\eta_\la = \ka_\la \sign(\inner{\phi_\la}{\hid_\star}) 
	& \text{ for } \la \in \supp(\hid_\star)
	\\
	 \eta_\la \in [-\ka_\la,\ka_\la] 
	 & \text{ for } \la \notin \supp(\hid_\star)  \,. 
 \end{cases}
\end{equation*}  
Since the family  $(\eta_\la)_{\la \in \La}$ is square summable, $\eta_\la = \pm \ka_\la $ can be obtained for only finitely many $\la$  and therefore  $\partial \snorm{\hid_\star}_{1,\kappa}$ is nonempty if and only if $\hid_\star$ is sparse.

\begin{remark}[Weighted $\ell^1$-norm]
For $\eta = \sum_{\la \in \La} \eta_\la \phi_\la \in \partial \snorm{\hid_\star}_{1,\kappa} $  define
\begin{align}
\Omega[\eta] &\coloneqq \{ \la \in \La \colon \abs{\eta_\la} = \ka_\la \}
\\
m[\eta] & \coloneqq \min \{ \ka_\la - \abs{\eta_\la} \colon \la \notin \Omega[\eta] \}.
\end{align}
Then $\Omega[\eta]$ is  finite and as $(\eta_\la)_{\la \in \La}$ converges to zero, $m[\eta]$ is well-defined with $m[\eta] > 0$.
Because $\snorm{\edot}_{1,\kappa}$ is positively homogeneous it  holds  $\snorm{\hid_\star}_1 = \inner{\eta}{\hid_\star}$.  Thus,  for $\hid \in \Hr$,
\begin{align}
\bregman_\eta^{\snorm{\edot}_{1,\kappa}}(\hid,\hid_\star) &= \snorm{\hid}_{1,\kappa} - \inner{\eta}{\hid} \notag \\
& = \sum_{\la \in \La} (\ka_\la \abs{\inner{\phi_\la}{\hid}}- \eta_\la \inner{\phi_\la}{\hid}) \notag \\
& \geq  \sum_{\la \in \La} (\ka_\la- \abs{\eta_\la}) \abs{\inner{\phi_\la}{\hid}} \notag \\
& \geq  m[\eta] \sum_{\la \notin \Omega[\eta]}  \abs{\inner{\phi_\la}{\hid}}  \,. \label{eq:lower}
\end{align}
Estimate \eqref{eq:lower} implies that if $\bregman_\eta^{\snorm{\edot}_{1,\kappa}}(\hid_\al^\delta,\hid_\star)$ linearly converge  to $0$, then so does  $\sum_{\la \notin \Omega[\eta]}  \abs{\inner{\phi_\la}{\hid^\delta_\al}}$.
\end{remark}

\begin{lemma} \label{normbound}
Let $\Om \subseteq \La$ be finite, $\Ao_\Om \colon \Hr_\Om \to \Y$ injective and $\hid_\star \in \Hr_\Omega$. Then, for all $\hid  \in \Hr$,
\begin{equation} \label{eq:normbound}
\snorm{\hid-\hid_\star} \leq \snorm{\Ao^{-1}_\Omega} \snorm{\Ao \hid - \Ao \hid_\star} 
+ (1+ \snorm{\Ao^{-1}_\Omega} \snorm{\Ao} )   \sum_{\la \notin \Omega} \abs{\inner{\phi_\la}{\hid}}  \,.
\end{equation}
\end{lemma}

\begin{proof}
Because $\Hr_\Omega$ is finite dimensional and $\Ao_\Om$ is injective, the  inverse $\Ao^{-1}_\Om \colon \ran(\Ao_\Om) \to \Hr_\Omega$ is well defined and  bounded. Consequently,  
\begin{align*}
\snorm{\hid-\hid_\star} & \leq  \snorm{\pi_\Om \hid - \hid_\star} + \snorm{\pi_{\La \setminus \Omega} \hid} \\
&\leq  \snorm{\Ao^{-1}_\Omega} \snorm{\Ao_\Om (\pi_\Om \hid - \hid_\star)} + \snorm{\pi_{\La \setminus \Omega} \hid} \\
&\leq  \snorm{\Ao^{-1}_\Omega} \snorm{\Ao(\hid - \hid_\star) - \Ao \pi_{\La \setminus \Omega} \hid} + \snorm{\pi_{\La \setminus \Omega} \hid} \\
&\leq  \snorm{\Ao^{-1}_\Omega} \snorm{\Ao \hid - \Ao \hid_\star} 
+ (1+ \snorm{\Ao^{-1}_\Omega} \snorm{\Ao} ) \snorm{\pi_{\La \setminus \Omega} \hid} 
\,.
\end{align*}
Bounding the $\ell^2$-norm   by the $\ell^1$-norm yields \eqref{eq:normbound}.   
\end{proof}

\begin{lemma}\label{lem:normbound2}
Let $\hid_\star \in \Hr$ be sparse, $\eta \in \partial \snorm{\hid_\star}_{1,\kappa}$ and assume that $\Ao_{\Omega[\eta]} \colon \Hr_{\Omega[\eta]} \to \Y$ is injective. Then, for $\hid \in \Hr$,
\begin{equation} \label{eq:normbound2}
\snorm{\hid-\hid_\star} \leq \snorm{\Ao^{-1}_{\Omega[\eta]}} \snorm{\Ao \hid - \Ao \hid_\star} 
+   \frac{1+ \snorm{\Ao^{-1}_{\Omega[\eta]}} \snorm{\Ao} }{m[\eta]} \bregman_\eta^{\snorm{\edot}_{1,\kappa}}(\hid,\hid_\star) \,.
\end{equation}
\end{lemma}

\begin{proof}
Follows form   \eqref{eq:normbound}, \eqref{eq:lower}. 
\end{proof}

\subsection{Relaxed $\ell^1$ co-regularization} \label{sec:relaxed}

First we derive linear rates for the relaxed model  $\relaxed_{\al,\data^\delta}$.  These results will be derived under the following  condition.

\begin{cond} \mbox{}\label{ass:relaxed}
\begin{enumerate}[label=(1.\arabic*),leftmargin=1cm]
\item\label{ass:relaxed1}  $(\signal_\star,\hid_\star, \data_\star) \in \X \times \Hr \times \Y$ with  $\Wo \signal_\star = \hid_\star$, $\Ao \hid_\star = \data_\star$.
\item\label{ass:relaxed2}   $\exists u \in \Hr \colon $  $\Wo^* u \in \partial \reg(\signal_\star)$ 
\item\label{ass:relaxed3}  $\exists v \in \Y \colon$ $\Ao^* v-u \in \partial \snorm{\hid_\star}_{1,\ka}$
\item\label{ass:relaxed4}   $\Ao_{\Omega[\Ao^*v-u]}$ is injective.
\end{enumerate}
\end{cond}

Conditions \ref{ass:relaxed2}, \ref{ass:relaxed3}  are source conditions  very commonly assumed in regularization theory.   From  \ref{ass:relaxed3}  it follows that $\hid_\star$ is sparse and contained in  $\Hr_{\Omega[\Ao^*v-u]}$.  Condition \ref{ass:relaxed4} is the restricted injectivity condition.  
 
\begin{remark}[Product formulation] \label{rem:product}
We introduce the operator $\Mall  \colon \X \times \Hr \to \Hr \times \Y$ and the functional $\qreg \colon \X \times \Hr \to [0, \infty] $,
\begin{align}\label{eq:product:fr1}
&\Mall  (\signal ,\hid)  \coloneqq   (\Wo \signal-\hid, \Ao \hid) 
\\ \label{eq:product:fr2}
 &\qreg  (\signal,\hid) \coloneqq  \reg(\signal)+\snorm{\hid}_{1,\ka}
\,.
\end{align}
Using these  notions, one  can rewrite the  relaxed co-regularization functional $\relaxed_{\al,\data^\delta}$   as   
\begin{equation*}
	\relaxed_{\al,\data^\delta} (\signal,\hid) = \frac{1}{2} \snorm{\Mall(\signal ,\hid) -(0,\data^\delta)}^2 + \al \qreg(\signal,\hid).
\end{equation*}
Because  $\Wo$ and $\Ao$ are linear and bounded, $\Mall$ is linear and bounded, too.  Moreover, since  $\reg$ and $\snorm{\edot}_{1,\ka}$ are  proper, convex and wlsc,  $\qreg$ has these properties, too.  The subdifferential $\partial \qreg (\signal_\star,\hid_\star)$ is  given by 
$\partial \qreg (\signal_\star,\hid_\star) = \partial \reg (\signal_\star) \times \partial \snorm{\hid_\star}_{1,\kappa}$.
The Bregman distance  with respect to $\xi = (\xi_1, \xi_2)$ is  given  by 
\begin{equation} \label{eq:product:bregman}
	\bregman_{(\xi_1, \xi_2)}^\qreg((\signal,\hid),(\signal_\star, \hid_\star))
	=
	\bregman_{\xi_1}^{\reg}(\signal-\signal_\star) + \bregman_{\xi_2}^{\enorm{}_{1,\kappa}} (\hid-\hid_\star)  \,.
\end{equation}
\ref{ass:relaxed2}, \ref{ass:relaxed3} can be written as $ \Mall^*(u, v) \in \partial \qreg (\signal_\star,\hid_\star)$.
\end{remark}
 
Here comes our main estimate for the relaxed model.

\begin{theorem}[Relaxed $\ell^1$ co-regularization] \label{thm:relaxed}
Let  Condition~\ref{ass:relaxed} hold and consider the parameter choice $\al = C \delta$ for $C > 0$. Then for all $\data^\delta \in \Y$ with  $\snorm{\data^\delta-\data_\star} \leq \delta$ and all $(\signal_\al^\delta,\hid_\al^\delta) \in \argmin \relaxed_{\al,\data^\delta}$ we have 
\begin{align}
\bregman_{\Wo^*u}^\reg (\signal_\al^\delta,\signal_\star) & \leq c_{(u,v)} \delta, \label{A:B_rate} \\
\snorm{\hid_\al^\delta-\hid_\star} & \leq d_{(u,v)} \delta \,,
\end{align}
where
\begin{align*}
& c_{(u,v)} \coloneqq (1+C\snorm{(u,v)})^2 /( 2C)\\
 & \begin{multlined}[0.8\linewidth] d_{(u,v)} \coloneqq 2 \snorm{\Ao^{-1}_{\Omega[\Ao^* v-u]}} (1+ C \snorm{(u,v)}) 
 +  \frac{1+ \snorm{\Ao^{-1}_{\Omega[\Ao^* v-u]}} \snorm{\Ao}}{ m[\eta]} c_{(u,v)}. \end{multlined}
\end{align*}
\end{theorem}

\begin{proof}
According to \ref{ass:relaxed3}, $\eta \coloneqq \Ao^* v -u  \in \partial\snorm{\hid_\star}_{1,\kappa}$, which implies that $\Omega[\eta]$ is  finite and  $\hid_\star \in \Hr_{\Omega[\eta]}$.
With  \ref{ass:relaxed4} and   Lemma \ref{lem:normbound2} we therefore get 
\begin{equation} \label{eq:relaxed-aux1}
\snorm{\hid_\al^\delta - \hid_\star} \leq \snorm{\Ao^{-1}_{\Omega[\eta]}} \snorm{\Ao \hid^\delta_\al - \data_\star}  
+   \frac{(1+ \snorm{\Ao^{-1}_{\Omega[\eta]}} \snorm{\Ao})  }{ m[\eta]} \bregman_{\eta}^{\snorm{\edot}_{1,\kappa}}(\hid^\delta_\al,\hid_\star)\,.
\end{equation}
Using the product formulation as in  Remark \ref{rem:product}, according to   \ref{ass:relaxed2}, \ref{ass:relaxed3}  the source condition $ \Mall^*(u, v) \in \partial \qreg (\signal_\star,\hid_\star)$ holds. By Lemma \ref{lem:var} and the choice $\al= C \delta$   we obtain
\begin{align*}
& \snorm{ \Mall(\signal_\al^\delta, \hid_\al^\delta) - (0,\data^\delta)}   \leq (1+ 2C \snorm{(u,v)}) \, \delta
\\
& \bregman_{\Mall^*(u,v)}^\qreg \left( (\signal_\al^\delta,\hid_\al^\delta), (\signal_\star,\hid_\star) \right)    \leq (1+C\snorm{(u,v)})^2   / ( 2C) \, \delta \,. 
\end{align*}
Using \eqref{eq:product:fr1}, \eqref{eq:product:fr2}, \eqref{eq:product:bregman} we obtain     
\begin{align*}
\snorm{\Ao \hid^\delta_\al -\data^\delta}  & \leq (1+ 2C \snorm{(u,v)}) \, \delta
\\
\snorm{\Wo \signal_\al^\delta - \hid^\delta_\al} &  \leq (1+ 2C \snorm{(u,v)}) \, \delta
\\
\bregman_{\eta}^{\snorm{\edot}_{1,\kappa}}(\hid^\delta_\al,\hid_\star)  &\leq (1+C\snorm{(u,v)})^2 \, \delta  / ( 2C)  
\\
\bregman_{\Wo^*u}^\reg (\signal_\al^\delta,\signal_\star)   &  \leq (1+C\snorm{(u,v)})^2 \, \delta  / ( 2C) \,.
\end{align*}
Combining this with \eqref{eq:relaxed-aux1} completes the proof. 
\end{proof}

If $\reg$ is totally convex, then convergence in the Bregman distance implies convergence in the norm \cite[Lemma 3.31]{scherzer2009variational}.  For example, for the standard penalty $\reg =  \snorm{\edot}^2/2$ from Theorem~\ref{thm:relaxed} one deduces the rate $\snorm{\signal_\al^\delta - \signal_\star} = \mathcal{O}( \sqrt{\delta})$.

\subsection{Strict $\ell^1$ co-regularization}

Next we analyze the strict approach \eqref{eq:strict}. We derive linear  convergence rates under  the  following condition.
 
\begin{cond} \mbox{}\label{ass:strict}
\begin{enumerate}[label=(2.\arabic*),leftmargin=1cm]
\item\label{ass:strict2}  $(\signal_\star, \data_\star) \in \X \times  \Y$ satisfies $\Ao \Wo \signal_\star = \data_\star$.
\item\label{ass:strict3}  $\exists \nu \in \Y \colon$  $\Wo^* \Ao^* \nu \in  \partial ( \reg + \partial \snorm{\Wo (\edot)}_{1,\kappa}) (\signal_\star)$

\item\label{ass:strict4} $\exists \xi \in \partial \reg (\signal_\star)$  $\exists  \etaw \in  \partial \snorm{\edot}_1 (\Wo \signal_\star) \colon $  $\Wo^* \Ao^* \nu = \xi + \Wo^* \etaw$ 

\item\label{ass:strict5} $\Ao_{\Omega[\etaw]}$ is injective.  
\end{enumerate}
\end{cond}

Condition \ref{ass:strict3} is a source condition for  the forward operator $ \Wo\Ao$ and  the regularization functional   $\reg + \partial \snorm{\Wo (\edot)}_{1,\kappa}$. Condition \ref{ass:strict4} assumes  the splitting of the subgradient $\Wo^* \Ao^* \nu = \xi + \Wo^* \etaw$ into subgradients $ \xi \in \partial \reg (\signal_\star)$  and  $  \Wo^* \etaw \in   \partial \snorm{\Wo(\edot)}_1 ( \signal_\star) $.  The assumption \ref{ass:strict5}  is the restricted injectivity.

\begin{theorem}[Strict $\ell^1$ co-regularization]\label{thm:strict}
Let  Condition \ref{ass:strict} hold and consider the parameter choice $\al = C \delta$ for  $C > 0$. Then for  $\snorm{\data^\delta-\data_\star} \leq \delta$ and $\signal_\al^\delta \in \argmin \strict_{\al,\data^\delta}$ we have 
\begin{align}
\bregman_{\xi}^\reg  (\signal_\al^\delta, \signal_\star) & \leq c_{(\nu, \etaw)} \delta \label{D_xi_r}
\\
\snorm{\Wo \signal_\al^\delta - \Wo \signal_\star} & \leq d_{(\nu, \etaw)} \delta \,,
\end{align}
with the constants
\begin{align*}
& c_{(\nu,\etaw)} \coloneqq (1+C\snorm{\nu})^2/(2C)\\
 &
\begin{multlined}[0.8\linewidth]
d_{(\nu, \etaw)} \coloneqq 2 \snorm{\Ao^{-1}_{\Omega[\etaw]}} (1+ C \snorm{\nu}) 
+   \frac{1+ \snorm{\Ao^{-1}_{\Omega[\etaw]}} \snorm{\Ao}}{m[\etaw]}c_{(\nu,\etaw)} \,.
\end{multlined}
\end{align*}
\end{theorem}

\begin{proof}
Condition~\ref{ass:strict} implies that  $\Omega[\etaw]$ finite and $\Wo \signal_\star \in \Hr_{\Omega[\etaw]}$. From Lemma \ref{lem:normbound2} we obtain
\begin{equation} \label{thm:strict:aux1}
\snorm{\Wo \signal_\al^\delta - \Wo \signal_\star} \leq \snorm{\Ao^{-1}_{\Omega[\etaw]}} \snorm{\Ao \Wo \signal_\al^\delta - \data_\star} 
 +   \frac{ (1+ \snorm{\Ao^{-1}_{\Omega[\etaw]}} \snorm{\Ao}) }{m[\etaw]} \bregman_{\etaw}^{\snorm{\edot}_1}(\Wo \signal_\al^\delta, \Wo \signal_\star).
\end{equation}
According to  \ref{ass:strict3} and Lemma~\ref{lem:var} applied with   $\Mall  = \Ao \Wo$ and $\qreg   = \reg + \snorm{\Wo (\edot)}_{1,\kappa}$ we obtain      
\begin{align} \label{thm:strict:aux2}
& \snorm{\Ao \Wo \signal_\al^\delta - \data^\delta} \leq  (1+ 2C \snorm{\nu}) \, \delta 
 \\ \label{thm:strict:aux3}
& \bregman_{\Wo^* \Ao^* \nu}^\qreg \left(\signal_\al^\delta,\signal_\star \right) 
  \leq (1+C\snorm{\nu})^2/(2C) \, \delta.
\end{align}
From \ref{ass:strict3}  we obtain $\bregman_{\Wo^* \Ao^* \nu}^\qreg = \bregman_{\etaw}^{\snorm{\edot}_1}(\Wo (\edot), \Wo (\edot)) + \bregman_{\xi}^{\reg}$. Together with  \eqref{thm:strict:aux1}, \eqref{thm:strict:aux2}, \eqref{thm:strict:aux2}  this show the claim. 
\end{proof}

\section{Necessary Conditions}\label{sec:Spec}

In this section we show that the source condition and restricted injectivity are not only sufficient but also necessary for linear convergence of relaxed $\ell^1$ co-regularization. In the following we  restrict ourselves  to the $\ell^1$-norm 
\begin{equation*}
	\snorm{\edot}_1 \coloneqq \snorm{\edot}_{1,1} = \sum_{\la \in \La}\abs{\inner{\phi_\la}{\edot}}.
\end{equation*}
We denote  by $\Mall$ and $\qreg$ the product  operator and  regularizer defined in \eqref{eq:product:fr1}, \eqref{eq:product:fr2}. We call $(\signal_\star, \hid_\star)$ a $\qreg$-minimizing solution of $ \Mall (\signal, \hid) = (0,\data_\star)$ if $\signal_\star \in \argmin\{\qreg(\signal) \mid \Mall (\signal, \hid) = (0,\data_\star)\}$. 
In this section we fix the following list of assumptions which  is slightly stronger than Assumption \ref{ass:main}.

\begin{ass} \mbox{} \label{ass:converse}
\begin{enumerate}[label=(B.\arabic*),leftmargin=1cm]
\item\label{ass:converse0} $\Wo \colon \X \to \Hr$ is  linear and bounded with dense range.
\item\label{ass:converse1} $\Ao \colon \Hr \to \Y$ is  linear and bounded.
\item\label{ass:converse5}  $\reg \colon \Hr \to [0,\infty]$ is  proper, strictly convex and wlsc. 
\item\label{ass:converse2}  $\La$ is  countable index set.
\item\label{ass:converse3}  $(\phi_\la)_{\la \in \La} \subseteq \Hr$ is an ONB of $\Hr$.
\item\label{ass:converse4} $\forall \la \in \La \colon \phi_\la \in \ran(\Wo)$.
\item\label{ass:converse6}  $\exists \signal \in  \X \colon \reg(\signal) + \snorm{\Wo \signal}_1  < \infty$.
\item  $\reg$ is Gateaux differentiable at $\signal_\star$ if $(\signal_\star,\hid_\star)$ is the unique $\qreg$-minimizing solution of $\Mall(\signal,\hid)=(0,\data_\star)$. \end{enumerate}
\end{ass}

Under Assumption~\ref{ass:converse}, the equation $\Mall(\signal,\hid)=(0,\data_\star)$  has a unique $\qreg$-minimizing solution.

\begin{cond} \mbox{}\label{ass:strong}
\begin{enumerate}[label=(3.\arabic*),leftmargin=1cm]
\item\label{ass:strong1}  $(\signal_\star,\hid_\star, \data_\star) \in \X \times \Hr \times \Y$ with $\Ao \hid_\star = \data_\star$, $\Wo \signal_\star = \hid_\star$.
\item\label{ass:strong2}  $\exists u \in \Hr \colon $  $\Wo^* u \in \partial \reg(\signal_\star)$ 
\item\label{ass:strong3}  $\exists v \in \Y \colon$ $\Ao^* v-u \in \partial \snorm{\hid_\star}_1$
\item\label{ass:strong4}  $\forall \la \notin \supp(\hid_\star) \colon$ $\abs{\inner{\phi_\la}{\Ao^*  v -  u}} < 1$
\item\label{ass:strong5}  $\Ao_{\supp[\hid_\star]}$ is injective.
\end{enumerate}
\end{cond}

\subsection{Auxiliary results}

Condition~\ref{ass:strong} is clearly stronger than Condition~\ref{ass:relaxed}. Below we will show that these  conditions are actually equivalent.   For that purpose we start  with  several lemmas. These results are in the spirit of \cite{Gr11}  where necessary conditions for standard $\ell^1$ minimization have been derived.    
 
 \begin{lemma} \label{lem:ness1}
Assume  that $(\signal_\star,\hid_\star) \in \X \times \Hr$ is the unique $\qreg$-minimizing solution of  $\Mall(\signal,\hid)=(0,\data_\star)$, let $u \in \Hr$ satisfy $\Wo^* u =  \partial \reg(\signal_\star)$ and assume that $\hid_\star$ is sparse. Then:
\begin{enumerate}[label=(\alph*)]
\item \label{item:ness1a}  The restricted mapping $\Ao_{\supp(\hid_\star)}$ is injective. 
\item \label{item:ness1b} For every finite set $\Omega_1$ with $\supp(\hid_\star) \cap \Omega_1 = \emptyset$ there exists $\theta \in \Y$ such that
\begin{align*}&\forall \la \in   \supp(\hid_\star) \colon    \inner{\phi_\la}{\Ao^* \theta - u} =  \sign(\inner{\phi_\la}{\hid_\star}) 
 \\ 
&\forall \la \in   \Omega_1 \colon   \abs{\inner{\phi_\la}{\Ao^* \theta - u}} <  1    \,.
\end{align*}
\end{enumerate}
\end{lemma}

\begin{proof} \mbox{} \ref{item:ness1a}:
Denote $\Omega \coloneqq \supp(\hid_\star)$. After possibly replacing some basis vectors by $-\phi_\la$, we may assume without loss of generality that $\sign(\inner{\phi_\la}{\hid_\star})=1$ for  $\la \in \Omega$. 
Since $(\signal_\star,\hid_\star)$ is the unique $\qreg$-minimizing solution of $\Mall(\signal,\hid)=(0,\data_\star)$, it follows that
\[ \qreg(\signal_\star,\hid_\star) < \qreg(\signal_\star + tx, \hid_\star + t \Wo x) \quad \]
for all $t\neq 0$ and  all  $\signal \in \X$ with $w \coloneqq \Wo \signal \in \ker(\Ao)\setminus \{0\}$.
Because $\Omega$ is  finite, the mapping
\begin{equation*}
t \mapsto \snorm{\hid_\star + t w}_1 
= \sum_{\la \in \Omega} \abs{\inner{\phi_\la}{\hid_\star}+ t \inner{\phi_\la}{w}} + \abs{t} \sum_{\la \notin \Omega} \abs{\inner{\phi_\la}{w}}
\end{equation*}
is piecewise linear.
Taking the one-sided directional derivative of $\qreg$ with respect to $t$, we have
\begin{align}
\notag 0 & < \lim_{t \downarrow 0} \frac{\qreg(\signal_\star + tx, \hid_\star + t w)- \qreg(\signal_\star,\hid_\star)}{t} \\
\notag & =  \lim_{t \downarrow 0} \frac{\snorm{\hid_\star + t w}_1- \snorm{\hid_\star}_1}{t} + \lim_{t \downarrow 0} \frac{\reg(\signal_\star + tx)-\reg(\signal_\star)}{t} \\
 & = \sum_{\la \in \Omega} \inner{\phi_\la}{w} + \sum_{\la \notin \Omega} \abs{\inner{\phi_\la}{w}} + \inner{\Wo^* u}{x}.\label{inj_ineq_A}
\end{align}
For the  last equality we used that $\inner{\phi_\la}{\hid_\star}=1$ for all $\la \in \Omega$, that $\reg$ is Gateaux differentiable and that $\Wo^* u = \partial \reg(\signal_\star)$. 
Inserting $-(\signal,w)$ instead of $(\signal,w)$ in  \eqref{inj_ineq_A}  we deduce
\begin{equation} \label{inj_wineq_A}
\sum_{\la \notin \Omega} \abs{\inner{\phi_\la}{w}} > \Big\vert \sum_{\la \in \Omega} \inner{\phi_\la}{w}  + \inner{u}{w} \Big\vert 
\end{equation}
for all $w \in (\ker(\Ao) \cap \ran(\Wo))\setminus \{0\}$.
In particular, 
\begin{equation}
\forall w \in (\ker(\Ao) \cap \ran(\Wo))\setminus \{0\} \colon \sum_{\la \notin \Omega} \abs{\inner{\phi_\la}{w}} > 0
\end{equation}
and consequentially  $\ker(\Ao) \cap \ran(\Wo) \cap \Hr_\Omega = \{0\}$. 
Because $\phi_\la \in \ran(\Wo)$ for all $\la \in \La$ and  $\Omega$ is finite, we have  $\Hr_\Omega \subseteq \ran(\Wo)$.  Therefore  $\ker(\Ao) \cap \Hr_\Omega = \{0\}$ which verifies that  $\Ao_\Omega$ is injective.

 \ref{item:ness1b}: Let $\Omega_1 \subseteq \La$ be finite with $\Omega \cap \Omega_1 = \emptyset$. 
Inequality \eqref{inj_wineq_A} and the finiteness of $\Omega \cup \Omega_1$ imply the existence of a constant $\mu \in (0,1)$ such that, for  $w \in \ker(\Ao) \cap \Hr_{\Omega \cup \Omega_1}$,
\begin{equation} \label{mu_ineq}
\mu \sum_{\la \in \Omega_1} \abs{\inner{\phi_\la}{w}} \geq \Big\vert \sum_{\la \in \Omega} \inner{\phi_\la}{w}  + \inner{u}{w} \Big\vert \,. 
\end{equation} 
Assume for the moment  $\xi \in \ran(\Ao^*_{\Omega \cup \Omega_1})$. Then $\xi = \Ao^*_{\Omega \cup \Omega_1} \theta$ for some $\theta \in \Y$.
Due to \ref{ass:converse3},  $\pi_{\Omega}$ is an orthogonal projection and the adjoint of the embedding $\mathrm{i}_{\Omega}$.  The identity
$\pi_{\Omega \cup \Omega_1} \circ \Ao^* = (\Ao \circ i_{\Omega \cup \Omega_1})^* = \Ao^*_{\Omega \cup \Omega_1} $ implies that 
\[ \forall \la \in \Omega \cup \Omega_1 \colon \inner{\phi_\la}{\xi} = \inner{\phi_\la}{\Ao^*_{\Omega \cup \Omega_1} \theta} = \inner{\phi_\la}{\Ao^* \theta} \,.\]
By assumption, $\Hr_{\Omega \cup \Omega_1}$ is finite dimensional and therefore $
\ran(\Ao^*_{\Omega \cup \Omega_1})=\ker(\Ao_{\Omega \cup \Omega_1})^\perp \subseteq \Hr_{\Omega \cup \Omega_1}$, where $(\edot)^\perp$ denotes the orthogonal complement in $\Hr_{\Omega \cup \Omega_1}$.
Consequently we have to show the  existence of $\xi \in (\ker(\Ao_{\Omega \cup \Omega_1})^\perp \subseteq \Hr_{\Omega \cup \Omega_1}$ with 
\begin{align} \label{new_Nes_A}
\begin{aligned}
\inner{\phi_\la}{\xi} &= 1 + u_\la & \forall \la &\in \Omega, \\
\inner{\phi_\la}{\xi} &\in (u_\la - 1, u_\la +1)  & \forall \la &\in \Omega_1,
\end{aligned}
\end{align}
where $u_\la \coloneqq \inner{\phi_\la}{u}$. 

Define the element $z \in \Hr_{\Omega \cup \Omega_1}$  by $\inner{\phi_\la}{z} = 1 + u_\la$ for $\la \in \Omega$ and $\inner{\phi_\la}{z} =u_\la$ for $\la \in \Omega_1$. If $z \in (\ker(\Ao))^\perp$, then we choose $\xi \coloneqq z$ and \eqref{new_Nes_A} is fulfilled.
If, on the other hand, $z \notin (\ker(\Ao))^\perp$, then $\dim(\ker(\Ao_{\Omega \cup \Omega_1})) \eqqcolon s \geq 1$ and there exists a basis $(w^{(1)}, \ldots, w^{(s)})$ of $\ker(\Ao_{\Omega \cup \Omega_1})$ such that 
\begin{align}\label{w_equality}
\begin{aligned}
1 & = \inner{z}{w^{(i)}} \\
& = \sum_{\la \in \Omega} (1+u_\la) \inner{\phi_\la}{w^{(i)}} + \sum_{\la \in \Omega_1} u_\la \inner{\phi_\la}{w^{(i)}} \\
& \begin{multlined}[0.8\linewidth] = \sum_{\la \in \Omega} \inner{\phi_\la}{w^{(i)}} + \sum_{\la \in \Omega \cup \Omega_1} u_\la \inner{\phi_\la}{w^{(i)}}  + \sum_{\la \notin \Omega \cup \Omega_1} u_\la \inner{\phi_\la}{w^{(i)}}  \end{multlined} \\
& = \sum_{\la \in \Omega} \inner{\phi_\la}{w^{(i)}} + \inner{u}{w^{(i)}} \quad \forall i \in \{1, \ldots, s \}
\end{aligned}
\end{align}
Consider now the constrained minimization problem on $\Hr_{\Omega_1}$
\begin{align} \label{min_problem}
\begin{aligned}
&\max_{\la \in \Omega_1}\abs{\inner{\phi_\la}{z'}}  \to \min \\
&\text{subject to } \inner{z'}{w^{(i)}}   = -1 \quad \text{ for } i \in \{1, \ldots, s \}.
\end{aligned}
\end{align}
Because of the equality $1 = \inner{z}{w^{(i)}}$, the admissible vectors $z'$ in \eqref{min_problem} are precisely those for which $\xi \coloneqq z + z' \in (\ker(\Ao_{\Omega \cup \Omega_1}))^\perp$. 
Thus, the task of finding $\xi$ satisfying \eqref{new_Nes_A} reduces to showing that the value of \eqref{min_problem} is strictly smaller that $1$.  
Note that the dual of the convex function $z' \mapsto \max_{\la \in \Omega_1}\abs{\inner{\phi_\la}{z'}}$ is the function
\begin{equation}
\max_{\Omega_1} \ni z' \mapsto 
\begin{cases}
0 & \text{ if } \sum_{\la \in \Omega_1} \abs{\inner{\phi_\la}{z'}} \leq 1 \\
+ \infty & \text{ if } \sum_{\la \in \Omega_1} \abs{\inner{\phi_\la}{z'}} > 1.
\end{cases}
\end{equation}
Recalling that $\inner{z'}{w^{(i)}} = \sum_{\la \in \Omega_1} \inner{\phi_\la}{z'}{\phi_\la}{w^{(i)}}$, it follows that the dual problem to \eqref{min_problem} is the following constrained problem on $\R^s$:
\begin{align} \label{dual_problem}
\begin{aligned}
& S(p) \coloneqq - \sum_{i=1}^s p_i  \to \min  \\
& \text{subject to } \sum_{\la \in \Omega_1} \Big\vert \sum_{i=1}^s p_i \inner{\phi_\la}{w^{(i)}} \Big\vert   \leq 1 \,.
\end{aligned}
\end{align}
From \eqref{w_equality} we obtain that
\begin{equation*}
\sum_{\la \in \Omega} \sum_{i=1}^s p_i \inner{\phi_\la}{w^{(i)}} + \sum_{i=1}^s p_i \inner{u}{w^{(i)}} = \sum_{i=1}^s p_i = - S(p)
\end{equation*}
for every $p \in \R^s$. Taking $w = \sum_{i=1}^s p_i w^{(i)} \in \ker(\Ao) \cap \Hr_{\Omega \cup \Omega_1}$, inequality \eqref{mu_ineq} therefore implies that for every $p \in \R^s$ there exist $\mu \in (0,1)$ such that
\begin{align*} 
&\mu \sum_{\la \in \Omega_1} \Big\vert \sum_{i=1}^s p_i \inner{\phi_\la}{w^{(i)}} \Big\vert 
\\ 
& \qquad \geq \Big\vert \sum_{\la \in \Omega} \sum_{i=1}^s p_i \inner{\phi_\la}{w^{(i)}} + \sum_{i=1}^s p_i \inner{u}{w^{(i)}} \Big\vert \\
&\qquad  = \Big\vert \sum_{i=1}^s p_i \Big\vert = \abs{S(p)} \,.
\end{align*} 
From \eqref{dual_problem} it follows that $\abs{S(p)} \leq \mu$ for every admissible vector $p \in \R^s$ for  \eqref{dual_problem}.  Thus the value of $S(p)$ in \eqref{dual_problem} is greater than or  equal to $-\mu$.
Since the value of the primal problem \eqref{min_problem} is the negative of the dual problem \eqref{dual_problem}, this shows that the value of \eqref{min_problem} is at most $\mu$. 
As $\mu \in (0,1)$, this proves that the value of \eqref{min_problem} is strictly smaller than $1$ and, as we have shown above, this proves assertion \eqref{new_Nes_A}. 
\end{proof}

\begin{lemma} \label{lem:converse1}
Assume  that $(\signal_\star,\hid_\star) \in \X \times \Hr$ is the unique $\qreg$-minimizing solution of  $\Mall(\signal,\hid)=(0,\data_\star)$
and suppose   $\Wo^* u \in \partial \reg(\signal_\star)$, $\Ao^* v-u \in \partial \snorm{\hid_\star}_1$ for some  $(u,v) \in \Hr \times \Y$. Then $(\signal_\star,\hid_\star)$ satisfies Condition \ref{ass:strong}.
\end{lemma}

\begin{proof}
The restricted injectivity condition \ref{ass:strong5} follows from Lemma \ref{lem:ness1}. Conditions \ref{ass:strong1}, \ref{ass:strong2} are satisfied according to assumption. Define now
\begin{align*}
\Omega_1 & \coloneqq \Omega[\Ao^* v-u] \setminus \supp(\hid_\star) \\
& = \{\la \in \La \setminus \supp(\hid_\star) \mid \abs{\inner{\phi_\la}{\Ao^* v-u}} =1 \}.
\end{align*}
Because $(\inner{\phi_\la}{\Ao^* v-u})_{\la \in \La} \subseteq \ell^2(\La)$, the set $\Omega_1$ is finite. Let $\theta \in \Y$ be as in Lemma \ref{lem:ness1}~\ref{item:ness1b} and set  
\begin{align*}
& \snorm{\theta}_\infty \coloneqq \sup\{ \abs{\inner{\phi_\la}{\Ao^* \theta -u}} \mid \la \in \La \}
\\
& a \coloneqq (1-m[\Ao^* v-u])/(2 \snorm{\theta}_\infty)  
\\
& \hat v \coloneqq (1- a ) v + a \theta \,.
\end{align*}
Note that $a \in (0,1/2 ]$. Then the following hold:  
\begin{itemize}
\item If $\la \in \supp(\hid_\star)$, then 
\begin{equation*}
\inner{\phi_\la}{\Ao^* \hat v - u} 
= (1-a)\inner{\phi_\la}{\Ao^* v - u}  + a \inner{\phi_\la}{\Ao^* \theta - u} 
= \sign(\inner{\phi_\la}{\hid_\star}).
\end{equation*} 
\item If  $\la \in \Omega_1$, then 
\begin{equation*}
\abs{\inner{\phi_\la}{\Ao^* \hat v - u}} 
\leq (1-a)\abs{\inner{\phi_\la}{\Ao^* v - u}}  + a \abs{\inner{\phi_\la}{\Ao^* \theta - u}} 
 < (1-a)+a = 1.
\end{equation*}
\item If  $\la \in \La \setminus ( \supp(\hid_\star) \cup \Omega_1)$, then
\begin{align*} 
&\abs{\inner{\phi_\la}{\Ao^* \hat v - u}} \\ 
&\qquad \leq (1-a) \, m[\Ao^* v - u] + a \snorm{\theta}_\infty \\
&\qquad \leq m[\Ao^* v - u] + (1-m[\Ao^* v - u])/2 \\
&\qquad = (1+m[\Ao^* v - u])/ 2 < 1 .
\end{align*} 
\end{itemize}
Consequently $(u,\hat v)$ satisfies  \ref{ass:strong3}, \ref{ass:strong4}. 
\end{proof}
 
\begin{lemma} \label{lem:ness2}
Let $(\delta_k)_{k \in \N} \in (0, \infty)^\N$ converge to $0$, $(\data_k)_{k \in \N} \in \Y^\N$ satisfy $\snorm{\data_k - \data_\star} \leq \delta_k$ and $\signal_k \in \argmin \relaxed_{\al,\data^\delta}$ with $\al_k \geq C \delta_k$ for  $C > 0$.  Then $\snorm{\signal_k -\signal_\star} \to 0$ and  $\snorm{\Mall \signal_k -\data_k}  = \mathcal{O} (\delta_k)$ as $k \to \infty$ imply ${\ran(\Mall^*) \cap \partial \qreg(\signal_\star) \neq \emptyset}$.
\end{lemma}

\begin{proof}
See   \cite[Lemma 4.1]{Gr11}. The proof given there also applies to our situation. 
\end{proof}

\subsection{Main result}

The following theorem  in the main results of this section  and shows that the source condition and restricted injectivity are in fact necessary for linear convergence.   

\begin{theorem}[Converse results] \label{thm:converse}
Let  $(\signal_\star,\hid_\star, \data_\star) \in \X \times \Hr \in \Y$ satisfy $\Ao \hid_\star = \data_\star$ and $\Wo \signal_\star = \hid_\star$ and let Assumption~\ref{ass:strong} hold. Then the following statements are
equivalent:
\begin{enumerate}[label=(\roman*)]
\item $(\signal_\star,\hid_\star,\data_\star)$ satisfies Condition \ref{ass:strong}. \label{thm:char1}
\item $(\signal_\star,\hid_\star,\data_\star)$ satisfies  Condition \ref{ass:relaxed}.\label{thm:char2}
\item\label{Tiii}
 $\exists  \xi \in \partial \reg(\signal_\star) $ $\forall C >0$   $ \exists c_1, c_2 > 0 \colon $
 For $\al = C \delta$, $\snorm{\data^\delta-\data_\star} \leq \delta$,  and  $(\signal_\al^\delta,\hid_\al^\delta) \in \argmin \relaxed_{\alpha,\data^\delta}$ we have   
\begin{align}
\bregman_{\xi}^\reg (\signal_\al^\delta,\signal_\star) & \leq c_1 \delta  \label{eq:D1}\\ 
\snorm{\hid_\al^\delta - \hid_\star} & \leq c_2 \delta \,. \label{Ubound_h}
\end{align}

\item $(\signal_\star,\hid_\star)$ is the unique $\qreg$-minimizing solution of $\Mall(\signal,\hid) = (0,\data_\star)$ and $\forall C >0$ $\exists c_3 , c_4> 0$ with
\begin{align}
\snorm{\Ao \hid^\delta_\al -\data_\star} & \leq c_3 \delta \label{rateA}\\
\snorm{\Wo \signal_\al^\delta -\hid^\delta_\al} & \leq c_4 \delta  \label{rateW}
\end{align}
for  $\snorm{\data-\data_\star} \leq \delta$,  $\al = C \delta$, $(\signal_\al^\delta,\hid_\al^\delta) \in \argmin \relaxed_{\al,\data^\delta}$.\label{Tiv}
\end{enumerate}
\end{theorem}

\begin{proof}
Item \ref{thm:char1} obviously implies Item \ref{thm:char2}.
The implication \ref{thm:char2} $\Rightarrow$ \ref{Tiii} has been shown in Theorem \ref{thm:relaxed}. 
The rate in \ref{Tiii} implies that $\hid_\star$ is the second component of every $\qreg$-minimizing solution of  $\Mall(\signal,\hid)=(0,\data_\star)$. As  $\reg$ is strictly convex,  \eqref{eq:D1} implies that  $(\signal_\star,\hid_\star)$ is the unique $\qreg$-minimizing solution of $\Mo (\signal, \hid) = (0,\data_\star)$.  The rate in \eqref{rateA} follows trivially from \eqref{Ubound_h}, since $\Ao$ is linear and bounded.  To prove  \eqref{rateW} we choose $u \in \Hr$ with  $\partial \reg(\signal_\star) = \Wo^* u$ and proceed similar as in the  proof of Lemma \ref{lem:var}. Because $(\signal_\al^\delta,\hid_\al^\delta) \in \argmin \relaxed_{\al,\data^\delta}$,  
\begin{multline*}
\snorm{\Wo \signal_\al^\delta  - \hid^\delta_\al}^2 + \snorm{ \Ao \hid^\delta_\al - \data^\delta }^2 + 2\al  \reg(\signal_\al^\delta) + 2\al  \snorm{\hid^\delta_\al}_1 \\ 
\leq \snorm{\hid_\star -\hid^\delta_\al}^2 + \snorm{ \Ao \hid^\delta_\al - \data^\delta }^2 + 2 \al \reg (\signal_\star) + 2 \al \snorm{\hid^\delta_\al}_1
\end{multline*}
and therefore 
\begin{equation*}
\snorm{\Wo \signal_\al^\delta -\hid^\delta_\al}^2 \leq (c_2 \delta)^2 + 2\al ( \reg(\signal_\star)-\reg(\signal_\al^\delta))  \,.
\end{equation*}
By the definition of the Bregman distance
\begin{align*}
&\reg(\signal_\star) - \reg(\signal_\al^\delta)  
\\ 
&\quad \leq -D_{\Wo^*u}^\reg(\signal_\al^\delta,\signal_\star) -\inner{u}{\Wo \signal_\al^\delta - \hid_\star}  \\
&\quad \leq -D_{\Wo^*u}^\reg(\signal_\al^\delta,\signal_\star) + \snorm{u}\snorm{\Wo \signal_\al^\delta - \hid_\star} \\
&\quad \begin{multlined}[0.7\linewidth] \leq -D_{\Wo^*u}^\reg(\signal_\al^\delta,\signal_\star) + \snorm{u}\snorm{\Wo \signal_\al^\delta - \hid^\delta_\al}
 + \snorm{u} c_2\delta. \end{multlined}
\end{align*}
Since the Bregman distance is nonnegative, it follows 
\begin{align*}
& \begin{multlined}[0.8\linewidth] 0 \geq \snorm{\Wo \signal_\al^\delta -\hid^\delta_\al}^2 - 2 \al \snorm{u} \snorm{\Wo \signal_\al^\delta -\hid^\delta_\al} 
-2 \al \snorm{u} c_2 \delta - (c_2 \delta)^2  \end{multlined} \\
& \; \, \begin{multlined}[0.8\linewidth] = \left(\snorm{\Wo \signal_\al^\delta -\hid^\delta_\al}+c_2 \delta\right)  
\cdot \left(\snorm{\Wo \signal_\al^\delta -\hid^\delta_\al} -2 \al \snorm{u} - c_2 \delta \right) \end{multlined}
\end{align*}
and hence $\snorm{\Wo \signal_\al^\delta -\hid^\delta_\al} \leq (2C \snorm{u} + c_2 ) \delta$. 

Now let \ref{Tiv} hold and $\snorm{\data-\data_k} \leq \delta_k$, $\delta_k \to 0$. Choose $\al_k=C \delta_k$ and  $(\signal_k,\hid_k) \in \argmin \relaxed_{\al_k,\data_k}$.
The uniqueness of $(\signal_\star,\hid_\star)$ implies $\snorm{(\signal_k,\hid_k) - (\signal_\star, \hid_\star)} \to 0$ as $k \to \infty$, see \cite[Prop. 7]{Gr08}. Moreover, 
\begin{equation*}
	\snorm{(\Wo \signal_k - h_k, \Ao h_k -\data_k)}  \leq \snorm{\Wo \signal_\al^\delta -h_k}  
	+ \snorm{\Ao h_k -\data_k}  \leq (c_4+c_3+1)\delta_k  \,.
\end{equation*}
Lemma \ref{lem:ness2} implies $\ran(\Mall^*) \cap \partial \reg(\signal_\star,\hid_\star) \neq \emptyset$, which means that there exists $(u,v) \in \Hr \times \Y$ such that $\Wo^* u \in \partial \reg(\signal_\star)$ and $\Ao^* v-u \in \partial \snorm{\hid_\star}_1$.  Proposition \ref{lem:converse1} finally shows that Condition \ref{ass:strong} holds, which concludes the proof.
\end{proof}

\subsection{Numerical example}

We consider recovering a function from CS measurements of its primitive. The aim of this elementary  example  is to point out possible  implementation of the two proposed models and supporting the linear error estimates. Detailed comparison  with other  methods and figuring out limitations of each method is subject  of future research.

\begin{figure}[htb!]
\centering
\includegraphics[width=0.4\columnwidth]{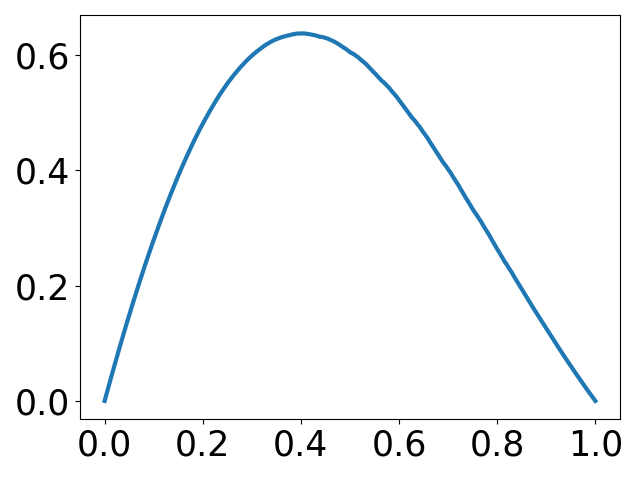} 
\includegraphics[width=0.4\columnwidth]{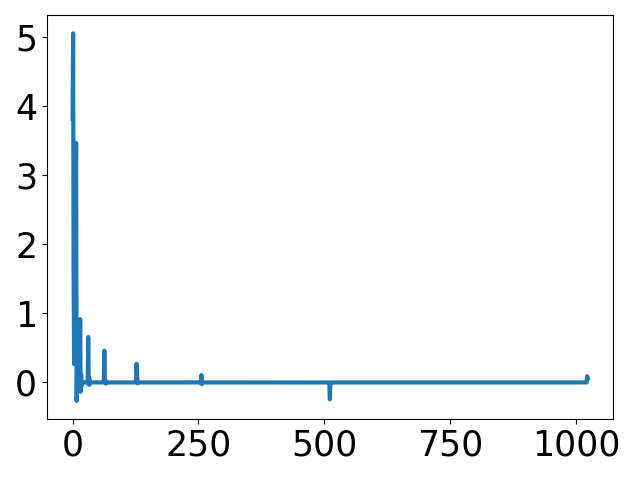}
\\[0.5em]
\includegraphics[width=0.4\columnwidth]{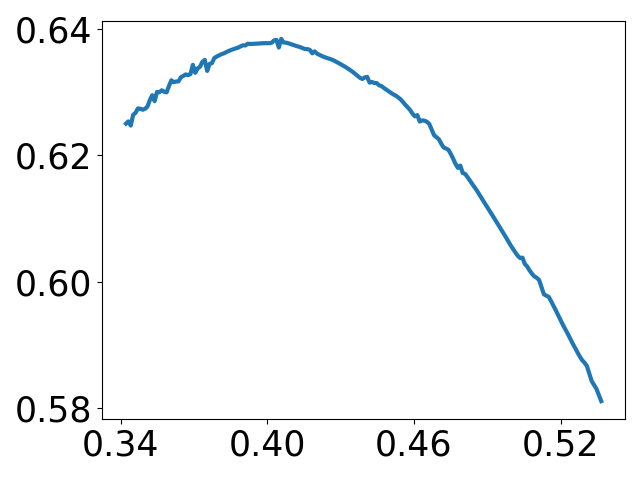} 
\includegraphics[width=0.4\columnwidth]{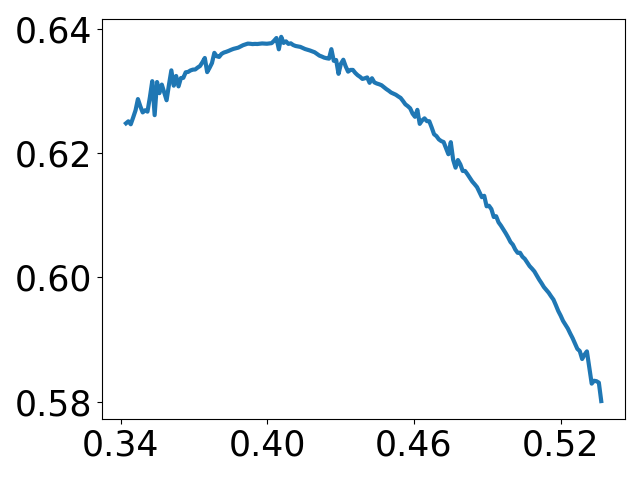} 
\\[0.5em]
\includegraphics[width=0.4\columnwidth]{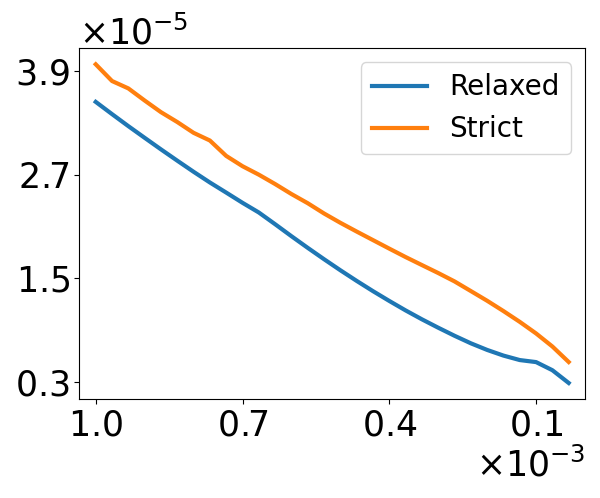}  
\includegraphics[width=0.4\columnwidth]{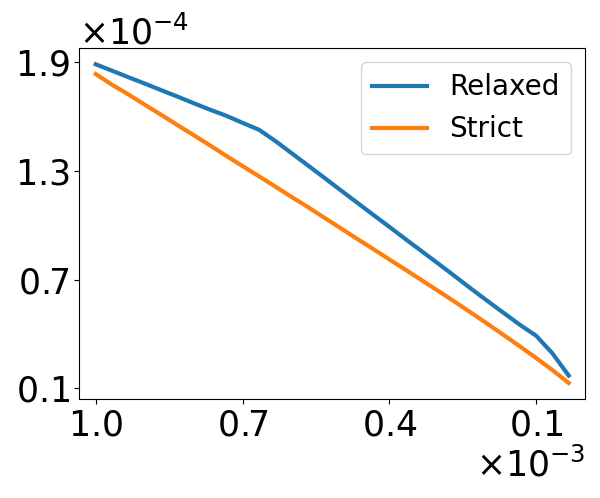}
\caption{Top: $\signal_\star$ (left) and $\Wo \signal_\star$ (right). Middle: Reconstruction using relaxed  (left) and  strict (right) $\ell^1$ co-regularization. Bottom:  $ \snorm{\signal_\star- \signal_\alpha^\delta}^2/2$ (left) and $\snorm{\hid_\star-\hid_\alpha^\delta}$ (right) as a functions of noise level $\delta$.}\label{fig:rec}
\end{figure}

The discrete operator  $\Wo \colon \R^N \to \R^N$ is taken as a discretization of the integration  operator $L^2[0,1] \to L^2[0,1] \colon f \mapsto  \int_0^t f$.  The CS  measurement matrix $\Ao \colon  \R^{m \times N}$   is taken  as random Bernoulli matrix with entries $0,1$.   We apply  strict  and relaxed $\ell^1$ co-regularization  with $\reg = \snorm{\edot}^2/2$, $(\phi_\la)_{\la \in \La}$   as Daubechies wavelet ONB with two vanishing moments and $\kappa_\la =1$. For minimizing the relaxed $\ell^1$ co-regularization functional we apply the Douglas-Rachford algorithm  \cite[Algorithm 4.2]{Pe11} and for strict $\ell^1$ co-regularization we apply the  ADMM algorithm \cite[Algorithm 6.4]{Pe11} applied to the constraint formulation  $ \argmin_{\signal, \hid} \{ \snorm{\Ao\Wo \signal - \data^\delta}^2/2 + \al \snorm{\signal}^2/2 + \al \snorm{\hid}_1 \mid  \Wo \signal = \hid \}$.

Results are shown in Figure \ref{fig:rec}.  The top row depicts the targeted signal $\signal_\star \in \R^N$ (left) for which  $\Wo \signal_\star$ is sparsely  represented by $(\phi_\la)_{\la \in \La}$ (right).  The  middle row shows reconstructions using the  strict and the relaxed co-sparse regularization from noisy data $\snorm{\data-\data_\star} \leq   10^{-5}$.  The bottom row plots $\bregman_{\signal_\star}^\reg(\signal_\al^\delta,\signal_\star) $ and  $\snorm{\hid^\delta_\al - \Wo \signal_\star}$ as  functions  of the noise  level. Note that  for  $\reg = \snorm{\edot}^2/2$ the Bregman distance  is given by $\bregman_{\signal_\star}^\reg(\signal_\al^\delta,\signal_\star) =  \snorm{\signal_\al^\delta - \signal_\star}^2/2$. Both error plots show a linear convergence rate  supporting Theorems~\ref{thm:relaxed},~\ref{thm:strict}, \ref{thm:converse}.

\section{Conclusion}

While the theory of CS  on direct data  is well developed, this by far not the case when compressed measurements are made on indirect data.  For that purpose, in this paper we study CS  from indirect data written as composite problem $\data^\delta = \Ao \Wo \signal_\star + z^\delta$ where $\Ao$ models the CS measurement operator and $\Wo$ the forward model generating indirect data and depending on the application at hand.  For signal reconstruction we have proposed two novel reconstruction methods, named relaxed and  strict $\ell^1$ co-regularization, for jointly estimating $\signal$ and $\hid_\star = \Ao \signal$. Note that  the  main conceptual difference  between the proposed  method over standard CS is that we use the $\ell^1$ penalty for indirect data $\Wo \signal_\star$ instead of  $\signal_\star$ together with another penalty for $\signal_\star$ accounting for the inversion of $\Ao$, and jointly recovering both unknowns.       

As main results for both reconstruction models we derive linear error estimates under source conditions and restricted injectivity (see Theorems~\ref{thm:strict},~\ref{thm:relaxed}).  Moreover, conditions have been shown to be even necessary to obtain  such results (see Theorem~\ref{thm:converse}).  Our results have been illustrated  on a simple numerical  example  for combined CS and numerical differentiation.  In future work further  detailed numerical  investigations are in order   comparing our models with standard CS approaches in practical important applications demonstrating strengths and limitations of different  methods. Potential  applications  include magnetic resonance imaging \cite{block2007undersampled,lustig2007sparse} or photoacoustic tomography \cite{sandbichler2015novel,provost2008application}.

\section*{Acknowledgment}

The presented work of A.E. and M.H. has been supported of the Austrian Science Fund (FWF), project P 30747-N32.


\begin{thebibliography}{10}

\bibitem{baraniuk2008simple}
R.~Baraniuk, M.~Davenport, R.~DeVore, and M.~Wakin.
\newblock A simple proof of the restricted isometry property for random
  matrices.
\newblock {\em Constructive Approximation}, 28(3):253--263, 2008.

\bibitem{block2007undersampled}
K.~T. Block, M.~Uecker, and J.~Frahm.
\newblock Undersampled radial {MRI} with multiple coils. iterative image
  reconstruction using a total variation constraint.
\newblock {\em Magnetic Resonance in Medicine: An Official Journal of the
  International Society for Magnetic Resonance in Medicine}, 57(6):1086--1098,
  2007.

\bibitem{burger2004convergence}
M.~Burger and S.~Osher.
\newblock Convergence rates of convex variational regularization.
\newblock {\em Inverse problems}, 20(5):1411, 2004.

\bibitem{candes2006robust}
E.~J. Cand{\`e}s, J.~Romberg, and T.~Tao.
\newblock Robust uncertainty principles: Exact signal reconstruction from
  highly incomplete frequency information.
\newblock {\em IEEE Transactions on information theory}, 52(2):489--509, 2006.

\bibitem{candes2006stable}
E.~J. Candes, J.~K. Romberg, and T.~Tao.
\newblock Stable signal recovery from incomplete and inaccurate measurements.
\newblock {\em Communications on Pure and Applied Mathematics: A Journal Issued
  by the Courant Institute of Mathematical Sciences}, 59(8):1207--1223, 2006.

\bibitem{Pe11}
P.~L. Combettes and J.~C. Pesquet.
\newblock {Proximal splitting methods in signal processing}.
\newblock {\em Springer Optimization and Its Applications}, 49:185--212, 2011.

\bibitem{donoho2006compressed}
D.~L. Donoho.
\newblock Compressed sensing.
\newblock {\em IEEE Transactions on information theory}, 52(4):1289--1306,
  2006.

\bibitem{flemming2018variational}
J.~Flemming.
\newblock {\em Variational Source Conditions, Quadratic Inverse Problems,
  Sparsity Promoting Regularization: New Results in Modern Theory of Inverse
  Problems and an Application in Laser Optics}.
\newblock Springer, 2018.

\bibitem{fourcat13CS}
S.~Foucart and H.~Rauhut.
\newblock {\em A mathematical introduction to compressive sensing}.
\newblock Applied and Numerical Harmonic Analysis. Birkh\"auser/Springer, New
  York, 2013.

\bibitem{fuchs2004sparse}
J.-J. Fuchs.
\newblock On sparse representations in arbitrary redundant bases.
\newblock {\em IEEE transactions on Information theory}, 50(6):1341--1344,
  2004.

\bibitem{Gr08}
M.~Grasmair, M.~Haltmeier, and O.~Scherzer.
\newblock {Sparse regularization with {$\ell^p$} penalty term}.
\newblock {\em Inverse Problems}, 24(5), 2008.

\bibitem{Gr11}
M.~Grasmair, O.~Scherzer, and M.~Haltmeier.
\newblock {Necessary and sufficient conditions for linear convergence of
  {$\ell^1$}-regularization}.
\newblock {\em Communications on Pure and Applied Mathematics}, 64(2):161--182,
  2011.

\bibitem{jin2009elastic}
B.~Jin, D.~A. Lorenz, and S.~Schiffler.
\newblock Elastic-net regularization: error estimates and active set methods.
\newblock {\em Inverse Problems}, 25(11):115022, 2009.

\bibitem{lorenz2008convergence}
D.~Lorenz.
\newblock Convergence rates and source conditions for tikhonov regularization
  with sparsity constraints.
\newblock {\em Journal of Inverse \& Ill-Posed Problems}, 16(5), 2008.

\bibitem{lustig2007sparse}
M.~Lustig, D.~Donoho, and J.~M. Pauly.
\newblock Sparse {MRI}: The application of compressed sensing for rapid mr
  imaging.
\newblock {\em Magnetic Resonance in Medicine: An Official Journal of the
  International Society for Magnetic Resonance in Medicine}, 58(6):1182--1195,
  2007.

\bibitem{provost2008application}
J.~Provost and F.~Lesage.
\newblock The application of compressed sensing for photo-acoustic tomography.
\newblock {\em IEEE transactions on medical imaging}, 28(4):585--594, 2008.

\bibitem{sandbichler2015novel}
M.~Sandbichler, F.~Krahmer, T.~Berer, P.~Burgholzer, and M.~Haltmeier.
\newblock A novel compressed sensing scheme for photoacoustic tomography.
\newblock {\em SIAM Journal on Applied Mathematics}, 75(6):2475--2494, 2015.

\bibitem{scherzer2009variational}
O.~Scherzer, M.~Grasmair, H.~Grossauer, M.~Haltmeier, and F.~Lenzen.
\newblock {\em Variational methods in imaging}.
\newblock Springer, 2009.

\bibitem{Zo05}
H.~Zou and T.~Hastie.
\newblock {Regularization and variable selection via the elastic net}.
\newblock {\em Journal of the Royal Statistical Society. Series B: Statistical
  Methodology}, 67(2):301--320, 2005.

\end{thebibliography}
\end{document}